\documentclass[runningheads]{llncs}

\usepackage{graphicx}
\usepackage{amsfonts}
\usepackage{amsmath}
\usepackage{ amssymb }
\usepackage{amsfonts}
\usepackage{mathtools}
\usepackage{ stmaryrd }
\usepackage{amscd}
\usepackage{bbold}
\usepackage{ amssymb }
\usepackage{mathtools}
\usepackage{mathpartir}
\usepackage[latin2]{inputenc}
\usepackage{t1enc}
\usepackage[mathscr]{eucal}
\usepackage{float}

\usepackage[switch]{lineno}
\usepackage{tikz-cd}
\usepackage{multicol}

\usepackage{ stmaryrd }

\usepackage{mdframed}
\usepackage{xcolor}
\usepackage{tcolorbox}

\definecolor{cornellred}{RGB}{196,18,48}
\definecolor{dartmouthgreen}{RGB}{0,112,60}
\definecolor{cisorange}{RGB}{247,147,33}
\definecolor{cisblue}{RGB}{67,199,244}
\definecolor{purple}{RGB}{96, 24, 168}
\definecolor{orange}{RGB}{198,79,0}

\newcommand{\cat}{%
\mathbf%
}

\newcommand{\sem}[1]{
  \left\llbracket #1 \right\rrbracket
}

\newcommand{\orange}[1]{
\textcolor{orange}{#1}
}

\newcommand{\purple}[1]{
\textcolor{purple}{#1}
}

\newcommand{\set}[2]{\{#1 \, | \, #2 \}}

\newcommand{\smid}{\ \mid \ }
\newcommand{\defined}{\coloneqq}
\newcommand{\subst}[3]{#1\{#2 / #3\}}

\newcommand{\nat}{\mathbb{N}}
\newcommand{\R}{\mathbb{R}}
\newcommand{\M}{\mathcal{M}}
\newcommand{\Pcal}{\mathcal{P}}

\newcommand{\lamb}[2]{\lambda #1. \ #2}
\newcommand{\app}[2]{#1 \ #2}
\newcommand{\letin}[3]{\mathsf{let}\ #1 = #2 \ \mathsf{in}\ #3}

\newcommand{\sample}[3]{\mathsf{sample}\ #1 \ \mathsf{as } \ #2 \ \mathsf{in}\ #3}
\newcommand{\lto}{\multimap}
\newcommand{\xto}{\xrightarrow}

\newcommand{\coin}{\mathsf{coin}}

\newenvironment{typing}
{\begin{center}
    \begin{tabular}{c c}
      \begin{mathpar}
}
{
   \end{mathpar}
  \end{tabular}
\end{center}
}

\usepackage{booktabs}   
\usepackage{subcaption} 

\begin{document}

\title{A Higher-Order Language for Markov Kernels and Linear Operators}
\titlerunning{}
%
\author{Pedro H. Azevedo de Amorim\inst{1}}
\authorrunning{P. H. Azevedo de Amorim}
%
\institute{Cornell University, Ithaca NY, USA}
%
%
%
\maketitle              
%

\begin{abstract}
Much work has been done to give semantics to probabilistic programming languages. In recent years, most of the semantics used to reason about probabilistic programs fall in two categories: semantics based on Markov kernels and semantics based on linear operators.

Both styles of semantics have found numerous applications in reasoning about probabilistic programs, but they each have their strengths and weaknesses. Though it is believed that there is a connection between them there are no languages that can handle both styles of programming.

In this work we address these questions by defining a two-level calculus and its categorical semantics which makes it possible to program with both kinds of semantics. From the logical side of things we see this language as an alternative resource interpretation of linear logic, where the resource being kept track of is sampling instead of variable use. 
\end{abstract}




\section{Introduction}

Probabilistic primitives have been a standard feature of programming
languages since the 70s. At first, randomness was mostly used to
program so called random algorithms, i.e. algorithms that require
access to a source of randomness. Recently, however, with the rise of
computational statistics and machine learning, randomness is also used
to program statistical models and inference algorithms.

Programming languages researchers have seen this rise in interest as
an opportunity to further study the interaction of probability and
programming languages, establishing it as an active subfield within
the PL community.


One of the main goals of this subfield is giving semantics to programming languages that are both expressive in the regular PL sense as well as in its abilities to program with randomness. One particular difficulty is that the mathematical machinery used for probability theory, i.e. measure theory, does not interact well with higher-order functions \cite{aumann1961}.

Currently, there are two classes of models of probabilistic programming --- in its broad sense --- that have found numerous applications: models based on linear logic and models based on Markov kernels. Since each kind of semantics has peculiarities that make them more or less adequate to give semantics to expressive programming languages, it is an important theoretical question to understand how these classes of models are related.

\subsubsection*{Linear Logic for Probabilistic Semantics}

The models of linear logic that have been used to give semantics to probabilistic languages are usually based on categories of vector spaces where programs are denoted by linear operators. We highlight two of them:
\begin{itemize}
\item Ehrhard et. al \cite{stablecones,lincone,pcoh} have defined
  models of linear logic with probabilistic primitives and have used
  the translation of intuitionistic logic into linear logic
  $A \to B = !A \lto B$, where $!A$ is the exponential modality, to
  give semantics to a stochastic $\lambda$-calculus.
\item Dahlqvist and Kozen \cite{roban} have defined an imperative,
  higher-order, linear probabilistic language and added a type constructor $!$ to
  accommodate non-linear programs.
\end{itemize}

The main advantage of models based on linear logic is that programs
are denoted by linear operators between spaces of distributions,
a formalism that has been extensively used to reason about
stochastic processes, as illustrated by Dahlqvist and Kozen who
have used results from ergodic theory to reason about a Gibbs
sampling algorithm written in their language, and by Clerc et al.
who have shown how Bayesian inference can be given semantics
using adjoint of linear operators \cite{clerc2017}.

Unfortunately, these insights are hard to realize in practice, since
languages based on linear logic enforce that variables must be used
exactly once, making it hard to use it as a programming language. The
usual way linear logic deals with this limitation is through the
$!$ modality which allows variables to be reused.

The problem with the exponential modality, when it comes to
probabilistic programming, is that they are usually difficult to
construct, do not have any clear interpretation in terms of
probability, making the linear operator formalism not applicable
anymore and, more operationally, through its connections with
call-by-name (CBN) semantics \cite{maraist1999}, makes it
mathematically hard to reuse sampled values.

Ehrhard et al. have found a way around this problem by introducing a
call-by-value (CBV) $\mathsf{let}$ operator that allows samples to be
reused \cite{stablecones,tasson2019}. In the discrete case this
operator is elegantly defined by a categorical argument which is
unknown to scale to the continuous case, which they deal with by
making use of an ad-hoc construction that is unclear if it can be
generalized to other models of linear logic.  Therefore, our current
understanding of models of linear logic does not provide a uniform way
of reusing samples.

The difference between CBV and CBN can be illustrated by the program
$\letin x {\mathsf{coin}} {x + x}$, where $\coin$ is a primitive that
outputs $0$ or $1$ with equal probability. In the CBN semantics each
use of $x$ corresponds to a new sample from $\coin$, whereas in the
CBV semantics the coin is only sampled once.

A subtler problem of probabilistic models based on linear logic is
that they are ill-equipped to program with joint distributions. For
instance, the language proposed by Ehrhard et. al can be easily
extended with product types which, under their semantics, would make
the type $\R \times \R$ be interpreted as $\M\R \times \M \R$, where
$\M\R$ is the set of distributions over $\R$ -- which is isomorphic to
the set of independent distributions over $\R^2$. Dahlqvist and Kozen
deal with this issue by adding primitive types $\R^n$ to their
language which are interpreted as the set of joint distributions over
$\R^n$. However, since they are not defined using the type
constructors provided by the semantic domain, programs of type $\R^n$
can only be manipulated by primitives defined outside the language.

\subsubsection*{Markov Kernel Semantics}

Markov kernels are a generalization of transition matrices,
i.e. functions that map states to probability distributions over
them. They are appealing from a programming languages perspective
because their programming model is usually captured by monads and
Kleisli arrows, a common abstraction in programming languages
semantics, and have been extensively used to reason about
probabilistic programs \cite{azevedo2019,scibior2018,barthe2014}. By
being related to monadic programming they differ from their linear
operator counterpart by being able to naturally capture a
call-by-value semantics which, as we argued above, is the most natural
one for probabilistic programming.

Unfortunately, even though these semantics can be generalized to
continuous distributions, they are notoriously brittle when it comes
to higher-order programming. Only recently, with the introduction of
quasi Borel spaces \cite{qbs} and its probability monad, it is
possible to give a kernel-centric semantics to higher-order
probabilistic programming with continuous distributions.

However, due to quasi Borel spaces being a different foundation to
probability theory, it is unclear which theorems and theories can be
generalized to higher-order. For instance, martingale theory has been
used in Computer Science to reason about termination of probabilistic
programs \cite{chakarov2013,mciver2017,huang2019}. In order to
generalize these ideas to higher-order functions it would be necessary
to define a quasi Borel version of martingales and prove appropriate
versions of the main theorems from martingale theory, a non-trivial
task.

\subsubsection*{Our Work: Combining both Kinds of Semantics}

Though both styles of semantics provide insights into how to interpret
probabilistic programming languages (PPL), it is still too early to
claim that we have a ``correct'' semantics which subsumes all of the
existing ones. Both approaches mentioned above have their
advantages and drawbacks.



In this work we shed some light into how both semantics relate to one
another by showing that it is possible to use both styles of semantics
to interpret a linear calculus that has higher-order functions, looser
linearity restrictions, a uniform way of dealing with sample reuse
and better syntax for programming joint distributions while still
being close to their kernel and linear operator
counterparts. Interestingly, we identify the joint distribution
problem described above to be a consequence of linear logic requiring
the non-linear product to be cartesian. In order to tackle this
problem we build on categorical semantics of linear logic and on
recent work on Markov categories, a suitable categorical
generalization of Markov kernels defined using semicartesian products.

We bridge the gap between these semantics by noting that the regular
resource interpretation of linear logic, i.e. $A \lto B$ being
equivalent to ``by using one copy of $A$ I get one copy of $B$'' is
too restrictive an interpretation for probabilistic programming.
Instead, we should think of usage as being equivalent to
sampling. Therefore the linear arrow $A \lto B$ should be thought of
as ``by sampling from $A$ once I get $B$'', which is the computational
interpretation of Markov kernels.
 
 
We realize this interpretation through a multilanguage approach: we
have one language that programs Markov kernels, a second language that
programs linear operators and add syntax that transports programs from
the former language into the latter one. To justify the viability of
our categorical framework we show how existing probabilistic semantics
are models to our language and show how, under mild conditions, this
semantics can be generalized to commutative effects.
 
%
%

 Our contributions are:
\begin{itemize}
\item[$\bullet$] We define a multi-language syntax that can program both Markov kernels as well as linear operators.(\textsection\ref{sec:syntax})
\item[$\bullet$] We define its categorical semantics and prove certain interesting equations satisfied by it. (\textsection\ref{sec:catsem})
\item[$\bullet$] We show that our semantics is already present in existing models for discrete and continuous probabilistic programming. (\textsection\ref{sec:concmod})
\item[$\bullet$] We show how our semantics can be generalized to commutative effects. (\textsection\ref{sec:commutative})
\end{itemize}

\section{Mathematical Preliminaries}

We are assuming that the reader is familiar with basic notions from category theory such as categories, functors and monads.
\subsection*{Probability Theory}
\label{sec:probtheory}

Transition matrices are one of the simplest abstractions used to model
stochastic processes. Given two countable sets $A$ and $B$, the entry
$(a, b)$ of a transition matrix is the probability of ending up in
state $b \in B$ whenever you start from the initial state $a \in A$
and every row adds up to $1$.

\begin{definition}
  The category $\cat{CountStoch}$ has countable sets as objects and
  transition matrices as morphisms. The identity morphism is the
  identity matrix and composition is given by matrix multiplication.
\end{definition}



Though transition matrices are conceptually simple, they can only model discrete probabilistic processes and, in order to generalize them to continuous probability we must use measurable sets and Markov kernels.

\begin{definition}
  A measurable set is a pair $(A, \Sigma_A)$, where $A$ is a set and $\Sigma_A \subseteq \mathcal{P}(A)$ is a $\sigma$-algebra, i.e. it contains the empty set and it is closed under complements and countable unions.
\end{definition}

\begin{definition}
  A function $f : (A, \Sigma_A) \to (B, \Sigma_B)$ is called measurable if for every $\mathcal{B} \in \Sigma_B$, $f^{-1}(\mathcal{B}) \in \Sigma_A$. 
\end{definition}

\begin{definition}
  Let $(A, \Sigma_A)$ be a measurable space. A probability distribution $(A, \Sigma_A)$ is a function $\mu : \Sigma_A \to [0, 1]$ such that $\mu(\emptyset) = 0$, $\mu(A) = 1$ and $\mu(\uplus_{i\in\nat} A_i) = \sum_{i\in\nat} \mu(A_i)$.
\end{definition}

Given two measurable sets $(A, \Sigma_A)$ and $(B, \Sigma_B)$ it is possible to define a $\sigma$-algebra over $A \times B$ generated by the sets $X \times Y$ which we denote by $\Sigma_A \otimes \Sigma_B$, where $X \in \Sigma_A$ and $Y \in \Sigma_B$. Furthermore, every pair of distributions $\mu_A$ and $\mu_B$ over $A$ and $B$ respectively, can be lifted to a distribution $\mu_A \otimes \mu_B$ over $A\times B$ such that $(\mu_A\otimes\mu_B)(X \times Y) = \mu_A(X)\mu_B(Y)$, for $X \in \Sigma_A$ and $Y \in \Sigma_B$.


\begin{definition}
  Let $(A, \Sigma_A)$ and $(B, \Sigma_B)$ be two measurable spaces. A Markov kernel is a function $f : A \times \Sigma_B \to [0, 1]$ such that

  \begin{itemize}
  \item For every $a \in A$, $f(a, -)$ is a probability distribution.
  \item For every $\mathcal{B} \in \Sigma_B$, $f(-, \mathcal{B})$ is a measurable function.
  \end{itemize}
\end{definition}
  
\begin{definition}
  The category $\cat{Kern}$ has measurable sets as objects and Markov kernels as morphisms. The identity arrow is the function $id_A(a,\mathcal{A}) = 1$ if $a \in \mathcal{A}$ and $0$ otherwise and Composition is given by $(f \circ g)(a, \mathcal{C}) = \int f(-, \mathcal{C}) d(g(a, -))$.
\end{definition}

\subsection*{Markov Categories}
The field of categorical probability was developed in order to get a
more conceptual understanding of Markov kernels. One of its
cornerstone definitions is that of a Markov category which are
categories where objects are abstract sample spaces, morphisms are
abstract Markov kernels and every object has ``contraction'' and
``weakening'' morphisms which correspond to duplicating and
discarding a sample, respectively, without adding any new randomness.

\begin{definition}[Markov category \cite{fritz2020}]
  A Markov category is a semicartesian symmetric monoidal category
  $(\cat{C}, \otimes, 1)$ in which every object $A$ comes equipped
  with a commutative comonoid structure, denoted by
  $\mathsf{copy}_X : X \to X \otimes X$ and
  $\mathsf{delete}_X : X \to 1$, where $\mathsf{copy}$ satisfies

  \[
    \mathsf{copy}_{X \otimes Y} = (id_X \otimes b_{Y, X} \otimes id_Y)
    \circ ( \mathsf{copy}_{X} \otimes \mathsf{copy}_{Y}),
  \]
\end{definition}

where $b_{Y, X}$ is the natural isomorphism
$Y \otimes X \cong X \otimes Y$. The category being semicartesian
means that the monoidal product comes equipped with projection
morphisms $\pi_1 : A \otimes B \to A$ and $\pi_2 : A\otimes B \to B$,
but it is not Cartesian because the equation
$(\pi_1\circ f, \pi_2 \circ f) = f$ does not hold in general which,
intuitively, corresponds to the fact that joint distributions might be
correlated.

\begin{theorem}[\cite{fritz2020}]
  $\cat{CountStoch}$ is a Markov category.
\end{theorem}

The monoidal product is given by the Cartesian product and the monoidal unit is the singleton set. The $\mathsf{copy}_X$ morphism is the matrix $X \times X \times X \to [0, 1]$ which is $1$ in the positions $(x, x, x)$ and $0$ elsewhere, and the $\mathsf{delete}_X$ morphism is the constant $1$ matrix indexed by $X$. 

\begin{theorem}[\cite{fritz2020}]
  $\cat{Kern}$ is a Markov category.
\end{theorem}

This category is the continuous generalization of $\cat{CountStoch}$ and the monoidal product is the Cartesian product with the \emph{product} $\sigma$-algebra and the monoidal unit is the singleton set $\{ * \} $. The $\mathsf{copy}_X$ morphism is the Markov kernel $\mathsf{copy}_X : X \times \Sigma_X \otimes \Sigma_X \to [0, 1]$ such that $\mathsf{copy}_X(x, S \times T) = 1$ if $x \in S \cap T$ and $0$ otherwise. Its $\mathsf{delete}$ morphism is simply the function that given any element in $X$, returns the function which is $1$ on the measurable set $\{ * \} $ and $0$ on the empty measurable set.  

\subsection*{Linear Logic and Monoidal Categories}
 We recall the categorical semantics of the multiplicative fragment of linear logic (MLL):

\begin{definition}[\cite{melliesll}]
  A category $\cat{C}$ is an MLL model if it is symmetric monoidal closed (SMCC), i.e. the functors $A \otimes -$ have a right adjoint $A \lto -$.
\end{definition}

We denote the monoidal product as $\otimes$ and the space of linear maps between objects $X$ and $Y$ as $X \lto Y$, $\mathsf{ev} : ((X \lto Y) \otimes X) \to Y$ is the counit of the monoidal closed adjunction and $\mathsf{cur} : \cat{C}(X \otimes Y, Z) \to \cat{C}(X, Y \lto Z)$ is the linear curryfication map. We use the triple $(\mathcal{C}, \otimes, \lto)$ to denote such models.


\begin{definition}
  Let $(\cat{C}, \otimes_{\cat{C}}, 1_{\cat{C}})$ and $(\cat{D}, \otimes_{\cat{D}}, 1_{\cat{D}})$ be two monoidal categories. We say that a functor $F : \cat{C} \to \cat{D}$ is \emph{lax monoidal} if there is a morphism $\epsilon : 1_{\cat{D}} \to F(1_{\cat{C}})$ and a natural transformation $\mu_{X, Y} : F(X) \otimes_{\cat{D}} F(Y) \to F(X \otimes_{\cat{C}} Y)$ making the diagrams in Figure~\ref{fig:lax} (in Appendix~\ref{app:diagrams}) commute.
\end{definition}
  
If $\epsilon$ and $\mu_{X, Y}$ are isomorphisms we say that $F$ is \emph{strong} monoidal.

One key observation of this paper is that there are many lax monoidal functors between Markov categories and models of linear logic that can interpret probabilistic processes.







\section{Syntax}
\label{sec:syntax}

In this section we will design a syntax that reflects the fact that linearity corresponds to
sampling, not variable usage. We achieve this by making use of a
multi-language semantics that enables the programmer to transport
programs defined in a Markov kernel-centric language (MK) to a linear,
higher-order, language (LL).

Our thesis is that in the context of probabilistic programming, linear
logic, through its connection with linear algebra, departs from its
usual Computer Science applications of enforcing syntactic invariants
and, instead, provides a natural mathematical formalism to express
ideas from probability theory, as shown by Dahlqvist and Kozen
\cite{roban}.

Therefore, since many probabilistic programming constructs, such as
Bayesian inference and Markov kernels, can be naturally interpreted in
linear logic terms, we believe that our calculus allows the user to
benefit from the insights linearity provides to PPL while unburdening
them from worrying about syntactic restrictions by making it possible
to also program using kernels.

We use standard notation from the literature: $\Gamma \vdash t : \tau$
means that the program $t$ has type $\tau$ under context $\Gamma$,
$\subst t x u$ means substitution of $u$ for $x$ in $t$ and
$\subst t {\overrightarrow{x}} {\overrightarrow{u}}$ is the
simultaneous substitution of the term list $\overrightarrow{u}$ for a
variable list $\overrightarrow{x}$ in $t$.

Both languages will be defined in this section and, for presentation's sake,
we are going to use orange to represent MK programs and purple to
represent LL programs. 

\subsection{A Markov Kernel Language}

We need a language to program Markov kernels. Since we
are aiming at generality, we are assuming the least amount of
structure possible. As such we will be working with the internal
language of Markov categories, as presented in Figure~\ref{fig:syntax}
and Figure~\ref{fig:typmk}\footnote{c.f. Appendix~\ref{app:sem}.}. Note that we are implicitly assuming a set
of primitives for the functions $\textcolor{orange}f$.

\begin{figure}[t]
  \begin{align*}
    \textcolor{orange}{
    \tau := 1 \smid \tau \times \tau}
  \end{align*}
  \begin{align*}
    & \textcolor{orange}{M, N} \defined \textcolor{orange}{x} \smid \textcolor{orange}{\mathsf{unit}} \smid \textcolor{orange}{\letin x M N} \smid \textcolor{orange}{(M, N)} \smid \textcolor{orange}{\pi_1 M} \smid \textcolor{orange}{\pi_2 N} \smid \textcolor{orange}{f(M)}
  \end{align*}
  \begin{align*}
    \textcolor{orange}{\Gamma} := \textcolor{orange}{\cdot} \smid \textcolor{orange}{x : \tau}, \textcolor{orange}{\Gamma}
  \end{align*}
  \caption{Syntax MK}
  \label{fig:syntax}
\end{figure}

By construction, every Markov category can interpret this language, as
we show in Figure~\ref{fig:semmk}, with types being interpreted as
\begin{align*}
  \sem{\orange 1} &= 1\\
  \sem{\orange \tau_1 \orange \times \orange \tau_2} &= \sem{\orange\tau_1} \times \sem{\orange\tau_2}
\end{align*}
and the contexts are interpreted using $\times$ over the
interpretation of the types. However, as it stands, it is not very
expressive, since it does not have any probabilistic primitives nor
does it have any interesting types since $1 \times 1 \cong 1$.

When working with concrete models (c.f. Section~\ref{sec:concmod}) we
can extend the language with more expressive types as well as with concrete
probabilistic primitives. For instance, in the context of continuous
probabilities we could add a $\R$ datatype and a
$\cdot \vdash \mathsf{uniform} : \R$ uniform distribution primitive.

Note that even though this language does not have any explicit
sampling operators, this is implicitly achieved by the $\mathsf{let}$
operator. For instance, the program
$\orange{\letin{x}{\mathsf{uniform}}{x + x}}$ samples from a uniform
distribution, binds the result to the variable $x$ and adds the sample
to itself.






\subsection{A Linear Language}

\begin{figure}[t]
  \begin{align*}
    \textcolor{purple}{\tau} := \textcolor{purple}{1} \smid \textcolor{purple}{\tau \lto \tau} \smid \textcolor{purple}{\tau \otimes \tau}
  \end{align*}
  \begin{align*}
    & \textcolor{purple}{t, u} \defined \textcolor{purple}{x} \smid \textcolor{purple}{\mathsf{unit}} \smid \textcolor{purple}{\lamb x t} \smid \textcolor{purple}{\app t u} \smid \textcolor{purple}{t \otimes u} \smid \textcolor{purple}{\letin {x\otimes y} t u}
  \end{align*}
  \begin{align*}
    \textcolor{purple}{\Gamma} := \textcolor{purple}{\cdot} \smid \textcolor{purple}{x : \tau}, \textcolor{purple}{\Gamma}
  \end{align*}
  \caption{Syntax LL}
  \label{fig:syntaxll}
\end{figure}

Our second language is a linear simply-typed $\lambda$-calculus, with the usual typing rules shown in Figure~\ref{fig:typell} in Appendix~\ref{app:sem}, which can
be interpreted in every symmetric monoidal closed category as shown in Figure~\ref{fig:semll}, also in Appendix~\ref{app:sem}, with types interpreted by
\begin{align*}
  \sem{\purple 1} &= 1\\
  \sem{\purple{\underline\tau}_1 \purple \otimes \purple{\underline\tau}_2} &= \sem{\purple{\underline\tau_1}} \otimes \sem{\purple{\underline\tau_2}}\\ 
  \sem{\purple{\underline\tau}_1 \purple \lto \purple{\underline\tau}_2} &= \sem{\purple{\underline\tau_1}} \lto \sem{\purple{\underline\tau_2}} 
\end{align*}
and the contexts are interpreted using $\otimes$ over the interpretation of the types.
Once again, we are aiming at generality instead of expressivity. 
In a concrete setting it would be fairly easy to extend the calculus
with a datatype $\purple \nat$ for natural numbers and probabilistic 
primitives such as $\purple\cdot \vdash \purple {\mathsf{coin}} : \purple \nat $ that flips
a fair coin.

The idea behind the particular linear logic models that we are interested in is that, by integration, Markov kernels can be seen as linear operators between vector spaces of probability distributions. As such, an LL program $\purple x : \purple \nat \vdash_{LL} \purple t : \purple \nat$ will be denoted by a linear function between distributions over the natural numbers. Therefore, from a programming point of view, variables are placeholders for probability distributions, i.e. computations, not values, and sampling occurs when variables are used.

\subsection{Combining Languages}
\label{sec:comblang}

The main drawback of the linear calculus above is that the syntactic
linearity restriction makes it hard to program with it, while the main
drawback of the Markov language is that it does not have
higher-order functions. In this section we will show how we can
combine both language so that we get a calculus with looser linearity
restrictions while still being higher-order.
\begin{figure}
  \begin{align*}
    &\textcolor{orange}{\tau} \defined \textcolor{orange}{1} \smid \textcolor{orange}{\tau \times \tau}\\
    &\textcolor{purple}{\underline{\tau}} \defined \textcolor{purple}{1} \smid \textcolor{purple}{\M} \textcolor{orange}{\tau} \smid \textcolor{purple}{\underline{\tau} \lto \underline{\tau}} \smid \textcolor{purple}{\underline{\tau} \otimes \underline{\tau}}
  \end{align*}
  \begin{align*}
    \textcolor{orange}{M, N} &\defined  \textcolor{orange}{x} \smid \textcolor{orange}{\mathsf{unit}} \smid \textcolor{orange}{\letin x M N} \smid \textcolor{orange}{f(M)} \\
    & \smid \textcolor{orange}{ (M, N)} \smid \textcolor{orange}{\pi_1 M} \smid \textcolor{orange}{\pi_2 M}\\
    \textcolor{purple}{t, u} & \defined \textcolor{purple}{x} \smid \textcolor{purple}{unit} \smid \textcolor{purple}{\lamb x t} \smid \textcolor{purple}{\app t u} \smid \textcolor{purple}{t \otimes u} \smid \textcolor{purple}{\letin {x \otimes y} t u} \\
    & |\, \sample {\textcolor{purple}{t_i}} {\textcolor{orange}{x_i}} {\textcolor{orange}{M}}
  \end{align*}
  \caption{Syntax LL+MK}
  \label{fig:syntaxmix}
\end{figure}

As we will show in Section~\ref{sec:concmod}, when looking at
concrete models for these languages we can see that the semantic
interpretations of variables in both languages are completely
different: in the MK language variables should be thought of as
values, i.e. the values that were sampled from a distribution, whereas
in the LL language, variables of ground type are distributions. In
order to bridge these languages we must use the observation that Markov
kernels --- i.e. open MK terms --- have a natural resource-aware
interpretation of being ``sample-once'' stochastic processes and, by
integration, can be seen as linear maps between measure spaces ---
i.e. open LL terms. The combined syntax for the language is depicted
in Figure~\ref{fig:syntaxmix}.

We now have a language design problem: we want to capture the fact
that every open MK program is, semantically, also an open LL term. The
naive typing rule is:
\begin{mathpar}
  \inferrule{x_1 : \tau_1, \cdots, x_n : \tau_n \vdash_{MK} M : \tau}{x_1 : \M \tau_1, \cdots, x_n : \M \tau_n \vdash_{LL} \mathsf{MK}(M):\M \tau}
\end{mathpar}
The problem with this rule is that it breaks substitution: the
variables in the premise are MK variables whereas the ones in the
conclusion are LL variables.


We solve this problem by making the syntax reflect a common
idiom of PPLs: compute distributions (elements of $\purple \M \orange\tau$),
sample from it and then use the result in a non-linear
continuation. This is captured by the following syntax:
\[
  \sample {\textcolor{purple}{t_1, \cdots, t_{n}}} {\textcolor{orange}{x_1, \cdots, x_{n}}} {\textcolor{orange}{M}}
\]
Note that we are sampling from LL programs $\purple t_i$ (possibly an empty
list), outputting the results to MK variables $\orange x_i$ and binding them
to an MK program $\orange M$. When clear from the context we simply use
$\sample {\purple t_i} {\orange x_i} {\orange M}$. Its corresponding typing rule is:
\begin{mathpar}

  \inferrule[Sample]{\textcolor{orange}{x_1} : \textcolor{orange}{\tau_1} \cdots \textcolor{orange}{x_n} : \textcolor{orange}{\tau_n} \vdash_{MK} \textcolor{orange}{M} : \textcolor{orange}{\tau} \quad \textcolor{purple}{\Gamma_i} \vdash_{LL} \textcolor{purple}{t_i} : \textcolor{purple}{\M} \textcolor{orange}{\tau_i} \\ 0 \leq i < n}{\textcolor{purple}{\Gamma_1}, \cdots, \textcolor{purple}{\Gamma_n} \vdash_{LL} \sample {\textcolor{purple}{t_i}} {x_i} {\textcolor{orange}{M}} : \textcolor{purple}{\M} \textcolor{orange}{\tau}}

\end{mathpar}

As the typing rule suggests, its semantics should be some sort of
composition. However, since we are composing programs that are
interpreted in different categories, we must have a way of translating
MK programs into LL programs --- as we will see in
Section~\ref{sec:catsem} this translation will be functorial. The
operational interpretation of this rule is that we have a set of
distributions $\{ \textcolor{purple}{t_i}\}$ defined using the linear
language --- possibly using higher-order programs --- we sample from
them, bind the samples to the variables $\{\textcolor{orange}{x_i}\}$
in the MK program $\textcolor{orange}{M}$ where there are no linearity
restrictions. Note that the rule above looks very similar to a monadic
composition, though they are semantically different
(cf. Section~\ref{sec:catsem}).

With this new syntax we can finally program in accordance with our new
resource interpretation of linear logic, allowing us to write the
program
\[ \sample {\textcolor{purple}{\coin}} {\textcolor{orange}{x}} {\textcolor{orange}{(x = x)}},\]
which flips a coin once and tests the result for equality with itself, making it equivalent to $\mathsf{true}$.

This combined calculus enjoys the expected syntactic properties\footnote{To avoid visually polluting the proofs we will drop the color code in Theorem~\ref{th:subst} and Theorem~\ref{th:comp}}.

\begin{theorem}
  \label{th:subst}
  Let $\Gamma, x : \underline{\tau_1} \vdash_{LL} t : \underline{\tau}$ and $\Delta \vdash_{LL} u : \underline{\tau_1}$ be well-typed terms, then $\Gamma, \Delta \vdash_{LL} \subst t x u : \underline{\tau}$
\end{theorem}
\begin{proof}
The proof can be found in Appendix~\ref{app:proof}.
\end{proof}

The following example illustrates how we can use the MK language to
duplicate and discard linear variables.

\begin{example}
  The program which samples from a distribution $\textcolor{purple}{t}$ and then returns a perfectly correlated pair is given by:
  \[    \cdot \vdash_{LL} \sample {\textcolor{purple}{t}} {\textcolor{orange}{x}} {\textcolor{orange}{(x, x)}} : \textcolor{purple}\M (\textcolor{orange}{\tau \times \tau})\]
  Similarly, the program that samples from a distribution $\textcolor{purple}{t}$ and does not use its sampled value is represented by the term
  \[
    \cdot \vdash_{LL} \sample {\textcolor{purple}{t}} {\textcolor{orange}{x}} {\textcolor{orange}{\mathsf{unit}}} : \textcolor{purple}\M \textcolor{orange}1
  \]
\end{example}

\begin{example}
  Suppose that we have a Markov kernel given by an open MK term $x : \nat \vdash M : \nat$. If we want to encapsulate it as a linear program of type $\M \nat \lto \M \nat$ we can write:
  \[
    \cdot \vdash_{LL} \textcolor{purple}{\lambda \, meas\ldotp} (\sample
    {\textcolor{purple}{meas}} {\textcolor{orange}{x}}
    {\textcolor{orange}{M}}) : \textcolor{purple}\M \textcolor{orange}{\nat} \lto \textcolor{purple}\M \textcolor{orange}{\nat}
  \]
\end{example}

\begin{example}
  \label{ex:roban}
  As we explain in the introduction, Dahlqvist and Kozen must add many
  primitives to their language to work around their linearity
  restrictions. For instance, in order to write projection functions
  $\R^n \to \R^m$, $n > m$ they must add projection primitives to the
  language.

  By having compositional type constructors that can represent joint
  distributions , i.e. $\purple\M(\orange{\tau \times \tau})$, it is possible to write the program
  $ \sample {\purple t} {\orange x} {(\app {\orange {\pi_1}} {\orange x}, \app {\orange {\pi_3}} {\orange x})}$ which samples from
  a distribution over triples and returns only the first and third
  components by only using the syntax of products in MK.
\end{example}
Unfortunately there are some aspects of this language that still are restrictive. For instance, imagine that we want to write an LL program that receives two ``Markov kernels'' $\textcolor{purple}\M \textcolor{orange}\nat \textcolor{purple}\lto \textcolor{purple}\M \textcolor{orange}\nat$ and a distribution over $\nat$ as inputs, samples from the input distribution, feeds the result to the Markov kernels, samples from them and adds the results. Its type would be
\[(\textcolor{purple}\M \textcolor{orange}\nat \textcolor{purple}\lto \textcolor{purple}\M \textcolor{orange}\nat) \textcolor{purple}\lto (\textcolor{purple}\M \textcolor{orange}\nat \textcolor{purple}\lto \textcolor{purple}\M \textcolor{orange}\nat) \textcolor{purple}\lto \textcolor{purple}\M \textcolor{orange}\nat \textcolor{purple}\lto \textcolor{purple}\M \textcolor{orange}\nat\]

Even though the program only requires you to sample once from each
distribution, it is still not possible to write it in the linear
language. 

We will show in Section~\ref{sec:catsem} how the type constructor
$\textcolor{purple}\M$ actually corresponds to an applicative functor
\cite{mcbride2008}, and the limitation above is actually a particular
case of a fundamental difference between programming with applicative
functors compared to programming with monads.

\begin{remark}
  We now have two languages that can interpret probabilistic
  primitives such as $\mathsf{coin}$. However, every primitive
  $\textcolor{orange}M$ in the MK language can be easily transported
  to an LL program by using an empty list of LL programs:
  $\sample \_ \_ {\textcolor{orange}M}$. Therefore it makes sense to
  only add these primitives to the MK language.
\end{remark}

\section{Categorical Semantics}
\label{sec:catsem}

As it is the case with categorical interpretations of languages/logics, types and contexts are interpreted as objects in a category and every well-typed program/proof gives rise to a morphism.

In our case, MK types $\textcolor{orange}{\tau}$ are interpreted as objects $\sem{\textcolor{orange}{\tau}}$ in a Markov category $(\cat{M}, \times)$ and well-typed programs $\textcolor{orange}{\Gamma} \vdash_{MK} \textcolor{orange}{M} : \textcolor{orange}{\tau}$ are interpreted as an $\cat{M}$ morphism $\sem{\textcolor{orange}{\Gamma}} \to \sem{\textcolor{orange}{\tau}}$, as shown in Figure~\ref{fig:semmk}. Similarly, LL types $\textcolor{purple}{\underline{\tau}}$ are interpreted as objects $\sem{\textcolor{purple}{\underline\tau}}$ in a model of linear logic $(\cat{C}, \otimes, \lto)$ and well-typed programs $\textcolor{purple}{\Gamma} \vdash_{LL} \textcolor{purple}{t} : \textcolor{purple}{\underline\tau}$ are interpreted as a $\cat{C}$ morphism $\sem{\textcolor{purple}{\Gamma}} \to \sem{\textcolor{purple}{\underline{\tau}}}$, as shown in Figure~\ref{fig:semll}.



  To give semantics to the combined language is not as straightforward. The sample rule allows the programmer to run LL programs, bind the results to MK variables and use said variables in an MK continuation. The implication of this rule in our formalism is that our semantics should provide a way of translating MK programs into LL programs. In category theory this is usually achieved by a functor $\M : \cat{M} \to \cat{C}$.

  However, we can easily see that functors are not enough to interpret the sample rule. Consider what happens when you apply $\M$ to an MK program $\textcolor{orange}{x} : \orange{\tau}_1, \orange{y} : \orange{\tau}_2 \vdash_{MK} \orange{N} : \orange {\tau}$:
  \[
    \M\sem{\textcolor{orange}{N}} : \M(\tau_1 \otimes \tau_2) \to \M\tau
  \]

  To precompose it with two LL programs outputting $\M \tau_1$ and $\M \tau_2$ we need a mediating morphism $\mu_{\tau_1, \tau_2} : \M\tau_1 \otimes \M\tau_2 \to \M(\tau_1 \times \tau_2)$. Furthermore, if $\textcolor{orange}{N}$ has three or more free variables, there would be several ways of applying $\mu$. Since from a programming standpoint it should not matter how the LL programs are associated, we require that $\mu_{\tau_1, \tau_2}$ makes the lax monoidality diagrams to commute. Therefore, assuming lax monoidality of $\mu$ we can interpret the sample rule:
  


  \[
    \inferrule[Sample]{\tau_1 \times \cdots \times \tau_n \xto{N} \tau \\ \Gamma_i \xto{t_i} \M \tau_i}{\Gamma \xto{t_1 \otimes \cdots \otimes t_n} \M \tau_1 \otimes \cdots \otimes \M \tau_n \xto{\mu} \M (\tau_1 \times \cdots \times \tau_n)\xto{\M N} \M \tau}
  \]

  In case it only has one MK variable, the semantics is given by $\sem{\textcolor{purple}{t}} ; \M\sem{\textcolor{orange}{N}}$ and in case it does not have any free variables the semantics is $\epsilon; \M\sem{\textcolor{orange}{N}}$.

    
    
    

    
    
    



  The equational theory of the LL languages is the well-known
  theory of the simply-typed $\lambda$-calculus and the MK equational
  theory has been described, in graphical notation, by Fritz \cite{fritz2020}.
  Something which is not obvious is understanding
  how they interact at their boundary. This is where $\M$ being a
  functor becomes relevant, since from functoriality it follows the
  two program equivalences:

\begin{theorem}
  Let $\textcolor{purple}{t}$, $\textcolor{orange}{M}$ and $\textcolor{orange}{N}$ be well-typed programs,
  \begin{align*}
    & \sem{\app {(\lamb {\textcolor{purple}{y}} {\sample {\textcolor{purple}{y}} z N})} {(\sample {\textcolor{purple}{t}} x M)}} = \\
    & \sem{\sample {\textcolor{purple}{t}} {\textcolor{orange}{x}} {(\textcolor{orange}{\letin y M N})}}
  \end{align*}
\end{theorem}
\begin{proof}
  \begin{align*}
    &\sem{\app {(\lamb {\textcolor{purple}{y}} {\sample {\textcolor{purple}{y}} {\textcolor{orange}{z}} {\textcolor{orange}{N}}})} {(\sample {\textcolor{purple}{t}} {\textcolor{orange}{x}} {\textcolor{orange}{M}})}} = \\
    &\sem{\textcolor{purple}{t}}; \M \sem{\textcolor{orange}{M}}; \M \sem{\textcolor{orange}{N}} = \sem{\textcolor{purple}{t}}; \M(\sem{\textcolor{orange}{M}}; \sem{\textcolor{orange}{N}}) = \\
    &\sem{\sample {\textcolor{purple}{t}} {\textcolor{orange}{x}} {(\textcolor{orange}{\letin y M N}})}
  \end{align*}
\end{proof}

\begin{theorem}
  Let $\textcolor{purple}{t}$ be a well-typed program,
  \[
    \sem{\sample {\textcolor{purple}{t}} {\textcolor{orange}{x}} {\textcolor{orange}{x}}} = \sem{\textcolor{purple}{t}}
  \]
\end{theorem}
\begin{proof}
$ \sem{\sample {\textcolor{purple}{t}} {\textcolor{orange}{x}} {\textcolor{orange}{x}}} = \sem{\textcolor{purple}{t}}; \M(\sem {\textcolor{orange}{x}}) = \sem {\textcolor{purple}{t}}; \M(id) = \sem {\textcolor{purple}{t}}; id = \sem{\textcolor{purple}{t}}$
\end{proof}

Furthermore, we also have a modularity property that can be easily proven:

\begin{theorem}
  Let $\textcolor{purple}{t}$, $\textcolor{orange}{M}$ and $\textcolor{orange}{N}$ be well-typed programs. If $\sem{\textcolor{orange}{M}} = \sem{\textcolor{orange}{N}}$ then
  \[
    \sem{\sample {\textcolor{purple}{t}} {\textcolor{orange}{x}} {\textcolor{orange}{M}}} = \sem{\sample {\textcolor{purple}{t}} {\textcolor{orange}{x}} {\textcolor{orange}{N}}}
  \]
\end{theorem}

The expected compositionality of the semantics also holds:

\begin{theorem}
  \label{th:comp}
  Let $x_1 : \tau_1, \cdots, x_n : \tau_n \vdash t : \tau$ and $\Gamma_i \vdash t_i : \tau_i$ be well-typed terms. $\sem{\Gamma_1, \cdots, \Gamma_n \vdash \subst t {\overrightarrow{x_i}} {\overrightarrow{t_i}} : \underline{\tau}} = (\sem{\Gamma_1 \vdash t_1 : \underline{\tau_1}} \otimes \cdots \otimes \sem{\Gamma_n \vdash t_n : \underline{\tau_n }}); \sem{\Gamma_1, \cdots, \Gamma_n } \vdash t : \underline{\tau}$.
\end{theorem}
\begin{proof}
  The proof can be found in Appendix~\ref{app:proof}.
\end{proof}
  \begin{mathpar}
    \inferrule[Subst]{\Gamma \vdash u_1 : \tau' \\ \Gamma \vdash u_2 : \tau' \\ \Gamma, x : \tau' \vdash t : \tau \\ \Gamma \vdash u_1 \equiv u_2 : \tau'}{\Gamma \vdash \subst t x {u_1} \equiv \subst t x {u_2} : \tau}
  \end{mathpar}
From this theorem we can conclude:

\begin{corollary}
  The Subst rule shown above is sound with respect to the categorical semantics.
\end{corollary}

Lax monoidal functors, under the name \emph{applicative functors}, are widely used in programming languages research\cite{mcbride2008}. They are often used to define embedded domain-specific languages (eDSL) within a host language. This suggests that from a design perspective the Markov kernel language can be thought of as an eDSL inside a linear language.

We have just shown that $\M$ being lax monoidal is sufficient to give semantics to our combined language, but what would happen if it had even more structure? If it were also full it would be possible to add a reification command\footnote{The proposed rule breaks the substitution theorem, but it is possible to define a variant for it where this is not the case.}:

\[
  \inferrule{\M\Gamma \vdash_{LL} t : \M \tau}{\Gamma \vdash_{MK} \mathsf{reify}(t) : \tau }
\]
 where $\M\Gamma$ is notation for every variable in $\Gamma$ being of the form $\M\tau'$, for some $\tau'$. The semantics for the rule would be taking the inverse image of $\M$. As we will show in the next section, there are some concrete models where $\M$ is full and some other models where it is not. Computationally, fullness of $\M$ can be interpreted as every program of type $\M \tau \lto \M \tau'$ being equal to a Markov kernel.

A property which is easier to satisfy is faithfulness, which is verified by both models in the next section. In this case the translation of the MK language into the LL language would be fully-abstract in the following sense:

\begin{theorem}
  Let $ \textcolor{orange}{x} : \textcolor{orange}{\tau_1} \vdash \textcolor{orange}{M} : \textcolor{orange}{\tau_2}$ and $ \textcolor{orange}{x} : \textcolor{orange}{\tau_1} \vdash\textcolor{orange}{ N} : \textcolor{orange}{\tau_2}$ be two well-typed MK programs. If $\M$ is faithful then $\sem{\sample {\textcolor{purple}{y}} {\textcolor{orange}{x} } \textcolor{orange}{M}} = \sem{\sample {\textcolor{purple}{y}} {\textcolor{orange}{x} } \textcolor{orange}{N}}$ implies $\sem {\textcolor{orange}{M}} = \sem {\textcolor{orange}{N}}$.
\end{theorem}
\begin{proof}
  $\sem{\sample {\textcolor{purple}{y}} {\textcolor{orange}{x} } \textcolor{orange}{M}} = \sem{\sample {\textcolor{purple}{y}} {\textcolor{orange}{x} } \textcolor{orange}{N}} \implies id_{\textcolor{purple}{\M}\textcolor{orange}{\tau_1}}; \M\sem{\textcolor{orange}{M}} = id_{\textcolor{purple}{\M}\textcolor{orange}{\tau_1}}; \M\sem{\textcolor{orange}{N}} \implies \sem{\textcolor{orange}{M}} = \sem{\textcolor{orange}{N}}$.
\end{proof}

\section{Concrete Models}
\label{sec:concmod}

In this section we show how existing models for both discrete as well as continuous probabilities fit within our formalism.

\subsection{Discrete Probability}

For the sake of simplicity we will denote the monoidal product of $\cat{CountStoch}$ as $\times$.

The probabilistic coherence space model of linear logic has been extensively studied in the context of semantics of discrete probabilistic languages\cite{pcoh}.

\begin{definition}[Probabilistic Coherence Spaces \cite{pcoh}]
  A probabilistic coherence space (PCS) is a pair $(|X|, \Pcal (X))$ where $|X|$ is a countable set and $\Pcal (X) \subseteq |X| \to \R^+$ is a set, called the \emph{web}, such that:

  \begin{itemize}
  \item
    $\forall a \in X\ \exists \varepsilon_a > 0\ \varepsilon_a \cdot
    \delta_a \in \mathcal P (X)$, where $\delta_a(a') = 1$ iff
    $a = a'$ and $0$ otherwise, and we use the notation
    $\varepsilon_a = \varepsilon(a)$;
  \item $\forall a \in X\ \exists \lambda_a\ \forall x \in \mathcal P (X)\ x_a \leq \lambda_a$;
  \item $\mathcal P (X)^{\perp\perp} = \mathcal P (X)$, where $\mathcal P (X)^\perp = \set {x \in X \rightarrow \R^+} {\forall v \in \mathcal P(X)\ \sum_{a \in X}x_av_a \leq 1}$.
  \end{itemize}

\end{definition}

We can define a category $\cat{PCoh}$ where objects are probabilistic coherence spaces and morphisms $X \lto Y$ are matrices $f : |X| \times |Y| \to \R^+$ such that for every $v \in \Pcal{(X)}$, $(f\, v) \in \Pcal{(Y)}$, where $(f \, v)_b = \sum_{a \in |A|}f_{(a,b)}v_a$.

\begin{definition}
  Let $(|X|, \Pcal{(X)})$ and $(|Y|, \Pcal{(Y)})$ be PCS, we define $X \otimes Y = (|X| \times |Y|, \set{x \otimes y}{ x \in \Pcal{(X)}, y \in \Pcal{(Y)}}^{\perp\perp})$, where $(x \otimes y)(a, b) = x(a)y(b)$
\end{definition}

\begin{lemma}
  \label{lem:natpcs}
  Let $X$ be a countable set, the pair $(X, \set{\mu : X \to \R^+}{\sum_{x \in X} \mu(x) \leq 1})$ is a PCS.
\end{lemma}
\begin{proof}
  The first two points are obvious, as the Dirac measure is a subprobability measure and every subprobability measure is bounded above by the constant function $\mu_1(x) = 1$.

  To prove the last point we use the --- easy to prove --- fact that $\Pcal{X} \subseteq \Pcal{X}^{\perp\perp}$. Therefore we must only prove the other direction. First, observe that, if $\mu \in \set{\mu : X \to \R^+}{\sum_{x \in X} \mu(x) \leq 1}$, then we have $\sum \mu(x)\mu_1(x) = \sum 1\mu(x) = \sum \mu(x) \leq 1$, $\mu_1 \in \set{\mu : X \to \R^+}{\sum_{x \in X} \mu(x) \leq 1}^\perp$.

  Let $\tilde{\mu} \in \set{\mu : X \to \R^+}{\sum_{x \in X} \mu(x) \leq 1}^{\perp\perp}$. By definition, $\sum \tilde{\mu}(x) = \sum \tilde{\mu}(x)\mu_1(x) \leq 1$ and, therefore, the third point holds.
\end{proof}

This lemma can be used to give semantics to probabilistic primitives. For instance, a fair coin is interpreted as a function $\coin : \nat \to [0,1]$ which is $.5$ at $0$ and $1$ and $0$ elsewhere and is an element of $\Pcal{(\nat)}$.

\begin{lemma}
  Let $X \to Y$ be a $\cat{CountStoch}$ morphism. It is also a $\cat{PCoh}$ morphism.
\end{lemma}

\begin{theorem}
  There is a lax monoidal functor $\M : \cat{CountStoch} \to \cat{PCoh}$.
\end{theorem}
\begin{proof}
  The functor is defined using the lemmas above. Functoriality holds due to the functor being the identity on arrows. The lax monoidal structure is given by $\epsilon = id_{1}$ and $\mu_{X, Y} = id_{X \times Y}$
\end{proof}

\begin{lemma}
  If $\mu \in \set{x \otimes y}{ x \in \M (X), y \in \M (Y)}^{\perp}$ then for every $x \in X$ and $y \in Y$, $\mu(x, y) \leq 1$.
\end{lemma}
\begin{proof}
If there were such indices such that $\mu(x_1, y_1) > 1$ then $\sum\sum \mu(x,y)(\delta_{x_1}\otimes\delta_{y_1})(x,y) > \mu(x_1, y_1) (\delta_{x_1}\otimes\delta_{y_1})(x_1,y_1) = \mu(x_1, y_1) > 1$, which is a contradiction. 
\end{proof}

\begin{lemma}
  Let $X$ and $Y$ be two countable sets, then
  \begin{align*}
    &\M X \otimes \M Y = \left( X \times Y, \set{\mu : X \times Y \to \R^+}{\sum_{x \in X} \sum_{y \in Y}\mu(x, y) \leq 1}\right ) = \\
    &\M(X \times Y).
  \end{align*}
\end{lemma}
\begin{proof}
By the lemma above it follows that if we have a joint probability distribution $\tilde{\mu}$ over $X \times Y$ and an element $\mu \in \set{x \otimes y}{ x \in \M (X), y \in \M (Y)}^{\perp}$ then $\sum \sum \mu(x, y)\tilde{\mu}(x,y) \leq \sum\sum \tilde{\mu}(x, y) \leq 1$.
\end{proof}

\begin{theorem}
  \label{thm:iso}
  Both $\epsilon$ and $\mu_{X, Y}$ are isomorphisms.
\end{theorem}
\begin{proof}
Since $\epsilon$ is the identity morphism, it is trivially an isomorphim. The morphisms $\mu_{X, Y}$ being an isomorphism is a direct consequence of the lemmas above.  
\end{proof}

\begin{theorem}
  \label{thm:full}
  The functor $\M$ is full.
\end{theorem}

Both results above can be directly used to enhance the syntax of the combined language. From Theorem~\ref{thm:iso} we can conclude that elements of type $\M(\tau_1 \times \tau_2)$, by projecting their marginal distributions, can be manipulated as if they had type $\M\tau_1 \otimes \M \tau_2$. Something to note is that when we do this marginalization process we lose potential correlations between the elements of the pair.

\subsection{Continuous Probability}


In order to accommodate continuous distributions we can use regularly ordered Banach spaces, whose detailed definition goes beyond the scope of this paper.

\begin{definition}[\cite{roban}]
  The category $\cat{RoBan}$ has regularly ordered Banach spaces as objects and regular linear functions as morphisms.
\end{definition}

\begin{theorem}
\label{th:laxcont}
  There is a lax monoidal functor $\M : \cat{Kern} \to \cat{RoBan}$.
\end{theorem}
\begin{proof}
  The functor acts on objects by sending a measurable space to the set of signed measures over it, which can be equipped with a $\cat{RoBan}$ structure. On morphisms it sends a Markov kernel $f$ to the linear function $\M( f)(\mu) = \int f d \mu$.

  The monoidal structure of $\cat{RoBan}$ satisfies the universal property of tensor products and, therefore, we can define the natural transformation $\mu_{X, Y} : \M(X) \otimes \M(Y) \to \M (X \times Y)$ as the function generated by the bilinear function $ \M(X); \M(Y) \lto \M (X \times Y)$ which maps a pair of distributions to its product measure. The map $\epsilon$ is, once again, equal to the identity function.

  The commutativity of the lax monoidality diagrams follows from the universal property of the tensor product: it suffices to verify it for elements $\mu_A \otimes \mu_B \otimes \mu_C$. 
\end{proof}

In $\cat{RoBan}$ the uniform distribution over the interval $[0,1]$ is an element of $\M \R$, meaning that it can soundly interpret a $\cdot \vdash_{LL}\mathsf{uniform} : \M \R$ primitive.

Even though $\M$ looks very similar to the discrete case, it follows from a well-known theorem from functional analysis that the functor is \emph{not} strong monoidal, meaning that there are joint probability distributions (elements of $\M(A \times B)$) that cannot be represented as an element of the tensor product $\M(A) \otimes \M(B)$ and, as such, programs of type $M(A \times B)$ must be manipulated in MK language, as shown in Example~\ref{ex:roban}.

\section{Beyond Probability}
\label{sec:commutative}
We have seen that this new resource interpretation is present in different models of linear logic models for probabilistic programming. In this section we show that this model can be generalized to commutative effects, i.e. effects where the program equation Commutativity below holds. Categorically, these effects are captured by monoidal monads\footnote{Monoidal monads are equivalent to commutative monads, which is the nomenclature usually used in the context of programming languages semantics.}. Due to length issues, we will not fully detail the definition of monoidal monads, but we suggest the interested reader to read Seal \cite{seal2013}.
\[
\inferrule[Commutativity]{\Gamma \vdash t_1 : \tau_1 \\ \Gamma \vdash t_2 :\tau_2 \\ \Gamma, x : \tau_1, y : \tau_2 \vdash u : \tau}{\letin{x_1}{t_1}{(\letin{x_2}{t_2}{u})} \equiv \letin{x_2}{t_2}{(\letin{x_1}{t_1}{u})} : \tau}
\]
\begin{definition}[\cite{seal2013}]
Let $(\cat{C}, \otimes, I)$ be a monoidal category and $(T, \eta,\mu)$ a monad over it. The monad $T$ is called monoidal if it comes equipped with a natural transformation $\kappa_{X, Y} : T X \otimes T Y \to T(X \otimes Y)$ making certain diagrams commute
\end{definition}

For probability monads the transformation $\kappa$ corresponds to forming the product probability distribution and, more generally, this can be thought of a program that runs both of its (effectful) inputs and pairs the outputs. 

Every monad give rise to the interesting categories $\cat{C}_T$ and $\cat{C}^T$ which are, respectively, the Kleisli category and Eilenberg-Moore category. The objects of $\cat{C}_T$ are the same as $\cat{C}$ and morphisms between $A$ and $B$ are $\cat{C}$ morphisms $A \to T B$, with the identity morphism being equal to the unit $\eta$ of the monad and composition is given by $f; g = f; T g ; \mu$.

The objects of the category $\cat{C}^T$ are pairs $(X, x)$, where $X$ is a $\cat{C}$ object and $x : T X \to X$ is a $\cat{C}$ morphism such that $\mu; x = T x ; x$ and $\eta; x = id_X$, and morphisms between objects $(X, x)$ and $(Y, y)$ are $\cat{C}$ morphisms $f : X \to Y$ such that $x ; f = T f ; y$.

For every monad $T$ there is a canonical inclusion functor $\iota : \cat{C}_T \to \cat{C}^T$ which maps $X$ to $(TX, \mu)$ and $f : X \to Y$ to $Tf ; \mu_Y$.

\begin{theorem}[\cite{borceux1994}]
The functor $\iota$ is full and faithful.
\end{theorem}

As we explain in Appendix~\ref{app:monad}, assuming enough structure
on the category $\cat{C}$ we can show that the triple
$(\cat{C}_T, \cat{C}^T, \iota)$ is a model to the MK+LL language and
we can bring our new resource interpretation of linear logic to other
commmutative effects.

An illustrative example is the powerset monad
$\mathcal{P} : \cat{Set} \to \cat{Set}$ which is monoidal and since
$\cat{Set}$ has the necessary structure, the triple
$(\cat{C}_{\mathcal{P}}, \cat{C}^{\mathcal{P}}, \mathcal{P})$ is a
model to our language and can be used to give semantics to
non-deterministic computation.

In the context of commutative effects other than randomness, the
syntax $\sample t x M$ does not make as much sense, in which case we
can use the syntax
$\mathsf{observe}\ t_i \ \mathsf{as } \ x_i \ \mathsf{in}\ M$
instead. Once again, operationally, the programs $t_i$ are fully
executed, the values are bound to $x_i$ in $M$ which is then executed.

Furthermore, other effects have other relevant effectful operations
and, therefore, we can assume that there is a set of operations in the
MK language that are interpreted in the Kleisli category and can be
transported to LL using $\mathsf{observe}$, similar to how it was done
in the probabilistic case.

For the non-deterministic case we can
assume the existence of typing rules for non-deterministic choice and
failure:
\begin{mathpar}
\inferrule[Choice]{\Gamma \vdash_{MK} t_1 : \tau \\ \Gamma \vdash_{MK} t_2 : \tau}{\Gamma \vdash_{MK} t_1 \oplus t_2 : \tau}
\and
\inferrule[Null]{~}{\Gamma \vdash_{MK} 0_{\tau} : \tau}
\end{mathpar}
satisfying the expected equations and interpreted using set-theoretic union and the empty set, respectively.


A similar connection between linear logic and monoidal monads has been
made by Benton and Wadler\cite{benton1996}, where they want to relate
Moggi's monadic $\lambda$-calculus with linear logic by showing
that if a monad is monoidal and the category has equalizers and
coequalizers, then the Eillenberg-Moore category is a model of linear
logic.

\section{Related Work}


\subsubsection*{Semantics of Probabilistic Programming}

Ehrhard et al. \cite{stablecones,lincone} have defined a model of linear logic $\cat{CLin}$ which can be used to interpret a higher-order probabilistic programming language. They have used the call-by-name translation of intuitionistic logic into linear logic $A \to B = !A \multimap B$ to give semantics to their language. The authors extend their language with a call-by-value $\mathsf{let}$ syntax which makes it possible to reuse sampled values. In order to give semantics to this new language they introduce a new category $\cat{CLin_m}$ which can interpret this new operator, at the cost of complicating their model. 

Because there is an analogous proof of Theorem~\ref{th:laxcont} with the category $\cat{CLin}$ replacing $\cat{RoBan}$, we can use their original, simpler, model to interpret our language, while not needing to use the linear logic exponential to interpret non-linear programs.

Dahlqvist and Kozen \cite{roban} have defined a category of partially ordered Banach spaces and shown that it is a model of intuitionistic linear logic. An important difference from their approach and the one mentioned above is that they embrace variable linearity as part of their syntax. As we argued in this paper, we believe that the syntactic restriction of linearity they have used is not adequate for the purposes of probabilistic programming. They deal with this limitation by adding primitives to their languages which, by using the results of Section~\ref{sec:concmod}, could be programmed using the MK language.

Quasi Borel spaces \cite{qbs} are a conservative extension of
$\cat{Meas}$ that are Cartesian closed and have a commutative
probability monad. The drawback of this model is that it is still not
as well understood as its measure-theoretic counterpart, and there are
theorems from probability theory used to reason about programs that
may not hold in the category of quasi Borel spaces $\cat{QBS}$.

Recently, Geoffroy \cite{geoffroy2021} has made progress in connecting
linear logic and quasi Borel Spaces by showing that a certain
subcategory of the Eillenberg-Moore category for the probability monad
in $\cat{QBS}$ is a model of classical linear logic, which we see as
an instance of our model where the MK language can have higher-order
functions as well.

\subsubsection*{Call-by-Push-Value}
The idea of having two distinct type systems that are connected by a
functorial layer is reminiscent of Call-by-Push-Value (CBPV)
\cite{levy2001}, which has a type system for values and a type system
for computations that are connected by an adjunction. In recent work,
Ehrhard and Tasson \cite{tasson2019} use the Eilenberg-Moore
adjunction of the linear logic exponential $!$ to give semantics to a
calculus that can interpret lazy and eager probabilistic computation,
allowing for the interpretation of an eager $\mathsf{let}$ operator
which is operationally similar to our $\mathsf{sample}$
construct. However, the existence of the $\mathsf{let}$ operator
depends on properties of the $!$ that are unknown to hold for
continuous distributions, while our semantics can naturally deal with
continuous distributions as we have shown in
Section~\ref{sec:concmod}.

Furthermore, the exponential which lies at the center of their approach is, semantically, hard to work with and does not have any clear connections to probability theory, making it unlikely that their semantics can be seen as a bridge between the Markov and linear semantics, which is the case for the models presented in Section~\ref{sec:concmod}.

Goubault-Larrecq \cite{goubault2019} has defined a CBPV domain semantics to a language that mixes probability and non-determinism, a long-standing challenge in the theory of programming languages. His focus is in understanding how to make probability interact with non-determinism in a sound way. He studies the full-abstraction of his semantics but does not deal with connections to linear logic.

\subsubsection{Acknowledgements}

The support of the National Science Foundation under grant CCF-2008083
is gratefully acknowledged. I would also like to thank Arthur Azevedo
de Amorim, Justin Hsu, Michael Roberts, Christopher Lam and Deepak
Garg for their useful comments on earlier versions of this paper.

 \bibliographystyle{splncs04}
\bibliography{mybib}

\begin{thebibliography}{10}
\providecommand{\url}[1]{\texttt{#1}}
\providecommand{\urlprefix}{URL }
\providecommand{\doi}[1]{https://doi.org/#1}

\bibitem{azevedo2019}
de~Amorim, A.A., Gaboardi, M., Hsu, J., Katsumata, S.y.: Probabilistic
  relational reasoning via metrics. In: Symposium on Logic in Computer Science
  (LICS) (2019)

\bibitem{aumann1961}
Aumann, R.J.: Borel structures for function spaces. Illinois Journal of
  Mathematics  (1961)

\bibitem{barthe2014}
Barthe, G., Fournet, C., Gr{\'e}goire, B., Strub, P.Y., Swamy, N.,
  Zanella-B{\'e}guelin, S.: Probabilistic relational verification for
  cryptographic implementations. In: Principles of Programming Languages (POPL)
  (2014)

\bibitem{benton1996}
Benton, N., Wadler, P.: Linear logic, monads and the lambda calculus. In:
  Symposium on Logic in Computer Science (LICS) (1996)

\bibitem{borceux1994}
Borceux, F.: Handbook of Categorical Algebra: Volume 2, Categories and
  Structures, vol.~2. Cambridge University Press (1994)

\bibitem{chakarov2013}
Chakarov, A., Sankaranarayanan, S.: Probabilistic program analysis with
  martingales. In: International Conference on Computer Aided Verification
  (CAV) (2013)

\bibitem{clerc2017}
Clerc, F., Danos, V., Dahlqvist, F., Garnier, I.: Pointless learning. In:
  International Conference on Foundations of Software Science and Computation
  Structures (FoSSaCS) (2017)

\bibitem{roban}
Dahlqvist, F., Kozen, D.: Semantics of higher-order probabilistic programs with
  conditioning. In: Principles of Programming Languages (POPL) (2019)

\bibitem{pcoh}
Danos, V., Ehrhard, T.: Probabilistic coherence spaces as a model of
  higher-order probabilistic computation. Information and Computation
  \textbf{209}(6),  966--991 (2011)

\bibitem{lincone}
Ehrhard, T.: On the linear structure of cones. In: Logic in Computer Science
  (LICS) (2020)

\bibitem{stablecones}
Ehrhard, T., Pagani, M., Tasson, C.: Measurable cones and stable, measurable
  functions: a model for probabilistic higher-order programming. In: Principles
  of Programming Languages (POPL) (2017)

\bibitem{fritz2020}
Fritz, T.: A synthetic approach to markov kernels, conditional independence and
  theorems on sufficient statistics. Advances in Mathematics  \textbf{370},
  107239 (2020)

\bibitem{geoffroy2021}
Geoffroy, G.: Extensional denotational semantics of higher-order probabilistic
  programs, beyond the discrete case (unpublished) (2021)

\bibitem{goubault2019}
Goubault-Larrecq, J.: A probabilistic and non-deterministic call-by-push-value
  language. In: Logic in Computer Science (LICS) (2019)

\bibitem{qbs}
Heunen, C., Kammar, O., Staton, S., Yang, H.: A convenient category for
  higher-order probability theory. In: Logic in Computer Science (LICS) (2017)

\bibitem{huang2019}
Huang, M., Fu, H., Chatterjee, K., Goharshady, A.K.: Modular verification for
  almost-sure termination of probabilistic programs. Proceedings of the ACM on
  Programming Languages (OOPSLA) (2019)

\bibitem{levy2001}
Levy, P.B.: Call-by-push-value. Ph.D. thesis (2001)

\bibitem{maraist1999}
Maraist, J., Odersky, M., Turner, D.N., Wadler, P.: Call-by-name,
  call-by-value, call-by-need and the linear lambda calculus. Theoretical
  Computer Science  (1999)

\bibitem{mcbride2008}
McBride, C., Paterson, R.: Applicative programming with effects. Journal of
  functional programming  \textbf{18}(1),  1--13 (2008)

\bibitem{mciver2017}
McIver, A., Morgan, C., Kaminski, B.L., Katoen, J.P.: A new proof rule for
  almost-sure termination. Proceedings of the ACM on Programming Languages
  (POPL) (2017)

\bibitem{melliesll}
Mellies, P.A.: Categorical semantics of linear logic. Panoramas et syntheses
  \textbf{27},  15--215 (2009)

\bibitem{scibior2018}
Scibior, A., Kammar, O., Vakar, M., Staton, S., Yang, H., Cai, Y., Ostermann,
  K., Moss, S., Heunen, C., Ghahramani, Z.: Denotational validation of
  higher-order bayesian inference. Proceedings of the ACM on Programming
  Languages  (2018)

\bibitem{seal2013}
Seal, G.J.: Tensors, monads and actions. arXiv preprint arXiv:1205.0101  (2012)

\bibitem{tasson2019}
Tasson, C., Ehrhard, T.: Probabilistic call by push value. Logical Methods in
  Computer Science  (2019)

\end{thebibliography}

\appendix

\section{Typing Rules and Denotational Semantics LL and MK}
\label{app:sem}

\begin{figure}[H]
\begin{mathpar}

  \inferrule[Var]{ }{\textcolor{orange}{\Gamma}, \textcolor{orange}{x} : \textcolor{orange}{\tau} \vdash \textcolor{orange}{x} : \textcolor{orange}{\tau}}

  \and
  
  \inferrule[Unit]{ }{\textcolor{orange}{\Gamma} \vdash \textcolor{orange}{\mathsf{unit}} : \textcolor{orange}{1}}

  \and
  
  \inferrule[Let]{\textcolor{orange}{\Gamma} \vdash \textcolor{orange}{M} : \textcolor{orange}{\tau_1} \\ \textcolor{orange}{\Gamma} , \textcolor{orange}{x} : \textcolor{orange}{\tau_1} \vdash \textcolor{orange}{N} : \textcolor{orange}{\tau}}{\textcolor{orange}{\Gamma} \vdash \textcolor{orange}{\letin x M N} : \textcolor{orange}{\tau}}

  \and

  \inferrule[Primitive]{\textcolor{orange}\Gamma \vdash \textcolor{orange}M : \textcolor{orange}{\tau_1} \\ \textcolor{orange}f :\textcolor{orange}{\tau_1} \to \textcolor{orange}{\tau_2}}{\textcolor{orange}\Gamma \vdash \textcolor{orange}{f(M)} : \textcolor{orange}{\tau_2}}

  \and

  \inferrule[Pair]{\textcolor{orange}{\Gamma} \vdash \textcolor{orange}{M} : \textcolor{orange}{\tau_1} \\ \textcolor{orange}{\Gamma} \vdash \textcolor{orange}{N} : \textcolor{orange}{\tau_2}}{\textcolor{orange}{\Gamma} \vdash \textcolor{orange}{(M, N)} : \textcolor{orange}{\tau_1 \times \tau_2}}

  \and

  \inferrule[Proj1]{\textcolor{orange}{\Gamma} \vdash \textcolor{orange}{M} : \textcolor{orange}{\tau_1 \times \tau_2}}{ \textcolor{orange}{\Gamma} \vdash \textcolor{orange}{\pi_1 M} : \textcolor{orange}{\tau_1}}
  
  \and

  \inferrule[Proj2]{\textcolor{orange}{\Gamma} \vdash \textcolor{orange}{M} : \textcolor{orange}{\tau_1 \times \tau_2}}{ \textcolor{orange}{\Gamma} \vdash \textcolor{orange}{\pi_2 M} : \textcolor{orange}{\tau_2}}

\end{mathpar}
\caption{Typing rules MK}
\label{fig:typmk}
\end{figure}

\begin{figure}
\begin{typing}
  \inferrule[Axiom]{ }{\textcolor{purple}{x} : \textcolor{purple}{\tau} \vdash \textcolor{purple}{x} : \textcolor{purple}{\tau}}

  \and

  \inferrule[Unit]{ }{\textcolor{purple}{\cdot} \vdash \textcolor{purple}{\mathsf{unit}} : \textcolor{purple}{1}}

  \and

  \inferrule[Abstraction]{\textcolor{purple}{\Gamma}, \textcolor{purple}{x} : \textcolor{purple}{\tau_1} \vdash \textcolor{purple}{t} : \textcolor{purple}{\tau_2}}{\textcolor{purple}{\Gamma} \vdash \textcolor{purple}{\lamb x t} : \textcolor{purple}{\tau_1 \lto \tau_2}}

  \and

  \inferrule[Application]{\textcolor{purple}{\Gamma_1} \vdash \textcolor{purple}{t} : \textcolor{purple}{\tau_1 \lto \tau_2} \\ \textcolor{purple}{\Gamma_2} \vdash \textcolor{purple}{u} : \textcolor{purple}{\tau_1}}{\textcolor{purple}{\Gamma_1}, \textcolor{purple}{\Gamma_2} \vdash \textcolor{purple}{\app t u} : \textcolor{purple}{\tau_2}}

  \and

  \inferrule[Tensor]{\textcolor{purple}{\Gamma_1} \vdash \textcolor{purple}{t} : \textcolor{purple}{\tau_1} \\ \textcolor{purple}{\Gamma_2} \vdash \textcolor{purple}{u} : \textcolor{purple}{\tau_2}}{\textcolor{purple}{\Gamma_1}, \textcolor{purple}{\Gamma_2} \vdash \textcolor{purple}{t \otimes u} : \textcolor{purple}{\tau_1 \otimes \tau_2}}

  \and

  \inferrule[LetTensor]{\textcolor{purple}{\Gamma_1}  \vdash \textcolor{purple}{t} : \textcolor{purple}{\tau_1 \otimes \tau_2} \\ \textcolor{purple}{\Gamma_2}, \textcolor{purple}{x} : \textcolor{purple}{\tau_1}, \textcolor{purple}{y} : \textcolor{purple}{\tau_2} \vdash \textcolor{purple}{u} : \textcolor{purple}{\tau}}{\textcolor{purple}{\Gamma_1}, \textcolor{purple}{ \Gamma_2} \vdash \textcolor{purple}{\letin {x\otimes y} t u} : \textcolor{purple}{\tau}}

\end{typing}
  \caption{Typing rules LL}
  \label{fig:typell}
\end{figure}

  \begin{figure}[H]

    \[
    \begin{tabular}{c}
      \begin{mathpar}

        \inferrule[Var]{ }{\Gamma \times \tau \xrightarrow{\mathsf{delete}\times id_\tau} 1 \times \tau \cong \tau}

        \and

        \inferrule[Pair]{\Gamma \xrightarrow{M} \tau_1 \\ \Gamma \xrightarrow{N} \tau_2}{\Gamma \xrightarrow{\mathsf{copy}} \Gamma \times \Gamma \xrightarrow{M \times N} \tau_1 \times \tau_2}

        \and

        \inferrule[Proj]{\Gamma \xrightarrow{M} \tau_1 \times \tau_2}{\Gamma \xrightarrow{M; (id_{\tau_1} \times \mathsf{delete})} \tau_1 \times 1 \cong \tau_1}

        \and

        \inferrule[Let]{\Gamma \xto{M} \tau_1 \\ \Gamma \times  \tau_1 \xto{N} \tau}{\Gamma \xto{\mathsf{copy}} \Gamma \times \Gamma \xto{(id_\Gamma \times M); N} \tau}

        \and

        \inferrule[Primitive]{\Gamma \xto{M} \tau_1 \\ \tau_1 \xto{f} \tau_2}{\Gamma \xto{M} \tau_1 \xto{f} \tau_2}

      \end{mathpar}
    \end{tabular}
    \]
    \caption{Denotational semantics for MK}
    \label{fig:semmk}
  \end{figure}

  \begin{figure}[H]

    \begin{tabular}{c}
      \begin{mathpar}

        \inferrule[Axiom]{ }{\tau \xrightarrow{id_\tau} \tau}

        \and

        \inferrule[Tensor]{\Gamma_1 \xto{t_1} \underline{\tau_1} \\ \Gamma_2 \xto{t_2} \underline{\tau_2}}{ \Gamma_1, \Gamma_2 \xto{t_1 \otimes t_2} \underline{\tau_1} \otimes \underline{\tau_2}}

        \and

        \inferrule[LetTensor]{\Gamma_1 \xto{t} \underline{\tau_1} \otimes \underline{\tau_2} \\ \Gamma_2 \otimes \underline{\tau_1} \otimes \underline{\tau_2} \xto{u} \underline{\tau}}{\Gamma_1 \otimes \Gamma_2 \xto{(id \otimes t); u} \underline{\tau}}

        \and

        \inferrule[Abstraction]{\Gamma \otimes  \underline{\tau_1} \xto{t} \underline{\tau_2}}{\Gamma \xto{\mathsf{cur}(\sem{t})} \underline{\tau_1} \lto \underline{\tau_2}}

        \and

        \inferrule[Application]{\Gamma_1 \xto{t} \underline{\tau_1} \lto \underline{\tau2} \\ \Gamma_2 \xto{u} \underline{\tau_1}}{\Gamma_1\otimes \Gamma_2 \xto{(t \otimes u); \mathsf{ev}} \tau_2}

      \end{mathpar}
    \end{tabular}
    \caption{Denotational semantics for LL}
    \label{fig:semll}
  \end{figure}

\section{Commutative Diagrams}
\label{app:diagrams}
\begin{figure}[H]
\begin{tikzcd}
	{(F(X) \otimes_{\cat D} F(Y)) \otimes_{\cat D} F(Z)} &&&&& {F(X) \otimes_{\cat D} (F(Y) \otimes_{\cat C} F(Z))} \\
	\\
	{F(X  \otimes_{\cat C} Y) \otimes_{\cat D} F(Z)} &&&&& {F(X) \otimes_{\cat D} F(Y \otimes_{\cat C} Z)} \\
	\\
	{F((X\otimes_{\cat C} Y) \otimes_{\cat C} Z)} & {} &&&& {F(X\otimes_{\cat C} (Y \otimes_{\cat C} Z))} \\
	&&&&& {}
	\arrow["{\mu \otimes id}", from=1-1, to=3-1]
	\arrow["\mu", from=3-1, to=5-1]
	\arrow["\alpha"', from=1-1, to=1-6]
	\arrow["{id \otimes \mu}"', from=1-6, to=3-6]
	\arrow["\mu"', from=3-6, to=5-6]
	\arrow["{F \alpha}", from=5-1, to=5-6]
      \end{tikzcd}

\begin{tikzcd}
	{1 \otimes_{\cat D} F(X)} &&& {F(1) \otimes_{\cat D} F(X)} && {F(X) \otimes_{\cat D} 1} &&& {F(X) \otimes_{\cat D} F(1)} \\
	\\
	{F(X)} &&& {F(1 \otimes_{\cat C} X)} && {F(X)} &&& {F(X \otimes_{\cat C} 1)}
	\arrow["{\epsilon \otimes id}", from=1-1, to=1-4]
	\arrow["{l^{\cat D}}"', from=1-1, to=3-1]
	\arrow["{F(l^{\cat C})}", from=3-4, to=3-1]
	\arrow["\mu", from=1-4, to=3-4]
	\arrow["{id\otimes \epsilon}", from=1-6, to=1-9]
	\arrow["\mu", from=1-9, to=3-9]
	\arrow["{F(r^{\cat D})}", from=3-9, to=3-6]
	\arrow["{r^{\cat D}}"', from=1-6, to=3-6]
      \end{tikzcd}
    \caption{Lax monoidal diagrams}
      \label{fig:lax}
\end{figure}

\section{Monoidal Monads and Their Algebras}
\label{app:monad}

An important theorem from the categorical probability literature is
that Markov categories are an abstraction of programming in the
Kleisli category of monoidal affine monads, where affinity means that
$T 1 \cong 1$.

\begin{theorem}[\cite{fritz2020}]
Let $(\cat{C}, \times, 1)$ be a cartesian category and $T : \cat{C} \to \cat{C}$ a monoidal (affine) monad. The Kleisli category $\cat{C}_T$ is a Markov category.
\end{theorem}

The monoidal product of $\cat{C}_T$ is $\times$ with unit $1$, the copy operation is given by $\Delta_X; \eta_X : X \to T (X \times X)$ and the deletion operation is given by $T 1 \cong 1$ and $1$ being terminal. 

Furthermore, under certain conditions, the Eilenberg-Moore category $\cat{C}^T$ for monoidal monads is symmetric monoidal closed. The monoidal unit is given by $T I$, the monoidal product is given by the coequalizer depicted in Figure~\ref{fig:emtensor} and the closed struture is given by the equalizer depicted in Figure~\ref{fig:emhom}.

\begin{theorem}
Let $\cat{C}$ be a symmetric monoidal closed category with equalizers, reflexive co-equalizers and $T : \cat{C} \to \cat{C}$ a monoidal monad. The category $\cat{C}^T$ is also symmetric monoidal closed.
\end{theorem}

\begin{figure}[]
    \centering
\[\begin{tikzcd}
	{T(TX \otimes TY)} && {TT(X \otimes Y)} && {T(X\otimes Y)} & {X \otimes_T Y}
	\arrow["T\kappa", from=1-1, to=1-3]
	\arrow["\mu", from=1-3, to=1-5]
	\arrow[bend right=15,swap, from=1-1, to=1-5]{drr}{T(x \otimes y)}
	\arrow[from=1-5, to=1-6]
\end{tikzcd}\]    
\caption{Symmetric Monoidal Structure in $\cat{C}^T$}
    \label{fig:emtensor}
\end{figure}

\begin{figure}[]
\centering
\[\begin{tikzcd}
	{X \multimap_T Y} & {X \multimap Y} && {T X \multimap T Y} && {T X \multimap Y}
	\arrow["s", from=1-2, to=1-4]
	\arrow["{id_{TX}\multimap y}", from=1-4, to=1-6]
	\arrow[bend right=15,swap, from=1-2, to=1-6]{drr}{x \multimap id_Y}
	\arrow[from=1-1, to=1-2]
\end{tikzcd}\]
\caption{Closed Structure in $\cat{C}^T$}
    \label{fig:emhom}
\end{figure}

Even though, in general, in order to define the monoidal product one requires a coequalizer, for our purposes we are only interested in products of the form $T A \otimes_T T B$ which, luckily, are easier to characterize, since the equality $TX \otimes_T TY = T(X \otimes Y)$ holds \cite{seal2013}.

In this case the lax monoidal transformations $\mu_{X, Y} : T X \otimes_T T Y \to T(X \otimes Y)$ and $\epsilon : F I \to F I$ are simply the identity morphisms. Besides, by using the universal properties of coequalizers it is possible to show the equality $\tilde{\alpha}_{TX, TY, TZ} = \alpha_{X, Y, Z}$, where $\tilde\alpha$ is the associator for the monoidal product $\otimes_T$.

\begin{theorem}
Let $\cat{C}$ be a symmetric monoidal category with reflexive co-equalizers and $T : \cat{C} \to \cat{C}$ a monoidal monad. The triple $(\iota, \mu, \epsilon)$ is a lax monoidal functor.
\end{theorem}
\begin{proof}
The proof follows by unfolding the definitions.
\end{proof}

\section{Proofs}
\label{app:proof}
Theorem~\ref{th:subst}.  Let $\Gamma, x : \underline{\tau_1} \vdash t : \underline{\tau}$ and $\Delta \vdash u : \underline{\tau_1}$ be well-typed terms, then $\Gamma, \Delta \vdash \subst t x u : \underline{\tau}$
\begin{proof}
  The proof follows by structural induction on the typing derivation $\Gamma, x : \underline{\tau_1} \vdash t : \underline{\tau}$:

  \begin{itemize}
  \item Axiom: Since $t = x$ then $\subst t x u = u$ and $\underline{\tau_1} = \underline{\tau}$.
  \item Abstraction: By hypothesis, $\Gamma, x : \underline{\tau_1}, y : \underline{\tau_2} \vdash t : \underline{\tau_3}$. Since we can assume wlog that $x \neq y$ and that $y \notin \Delta$, $\subst {\lamb y t} x u = \lamb y {\subst t x u}$. Therefore we can show that $\Gamma, \Delta \vdash \lamb y {\subst t x u} : \tau_2 \lto \tau_3$ by applying the rule Abstraction and by the induction hypothesis.
  \item Application: $\subst {\app {t_1} {t_2}} x u = \app {\subst {t_1} x u} {\subst {t_2} x u}$. Since the language LL is linear, only one of $t_1$ or $t_2$ will have $x$ as a free variable. By symmetry we can assume that $t_1$ has $x$ as a free variable and we can prove $\Gamma, \Delta \vdash \app {\subst {t_1} x u} {t_2} : \underline{\tau}$ by applying the rule Application and by the induction hypothesis.
    \item Sample: It is easy to prove that $\subst {(\sample t y M)} x u = \sample {(\subst t x u)} y M$
  \end{itemize}
\end{proof}

Theorem~\ref{th:comp}.  Let $x_1 : \tau_1, \cdots, x_n : \tau_n \vdash t : \tau$ and $\Gamma_i \vdash t_i : \tau_i$ be well-typed terms. $\sem{\Gamma_1, \cdots, \Gamma_n \vdash \subst t {\overrightarrow{x_i}} {\overrightarrow{t_i}} : \underline{\tau}} = (\sem{\Gamma_1 \vdash t_1 : \underline{\tau_1}} \otimes \cdots \otimes \sem{\Gamma_n \vdash t_n : \underline{\tau_n }}); \sem{\Gamma_1, \cdots, \Gamma_n } \vdash t : \underline{\tau}$.
\begin{proof}
  The proof follows by induction on the typing derivation of $t$.
    \begin{itemize}
  \item Axiom: Since $t = x$ then $\subst t x {t_0} = {t_0}$ and $\sem{\subst t x {t_0}} = \sem{t_0} = \sem{t_0}; id = \sem{t_0}; \sem{x}$.
  \item Unit: Since $t = x$ then $\subst t x {t_0} = {t_0}$ and $\sem{\subst t x {t_0}} = \sem{t_0} = \sem{t_0}; id = \sem{t_0}; \sem{x}$.
  \item Tensor: We know that $t = t_1 \otimes t_2$. Furthermore, from linearity we know that each free variable appears either in $t_1$ or in $t_2$. Without loss of generality we can assume that $\subst {(t_1 \otimes t_2)} {x_1, \cdots , x_n} {u_1,\cdots , u_n} = (\subst {t_1} {x_1, \cdots, x_k} {u_1, \cdots, u_k}) \otimes (\subst {t_2} {x_{k+1}, \cdots, x_n} {u_{k+1}, , \cdots, u_n})$. We can conclude this case from the induction hypothesis and functoriality of $\otimes$. 
  \item LetTensor: This case follows from the functoriality of $\otimes$ and the induction hypothesis.
  \item Abstraction: This case follows from unfolding the definitions, using the induction hypothesis and by naturality of $\mathsf{cur}$.
  \item Application: Analogous to the Tensor case
  \item Sample: This case is analogous to the Tensor case.
  \end{itemize}

\end{proof}

\end{document}